\documentclass[conference]{IEEEtran}

\usepackage{color}
\usepackage{cite}
\usepackage{fixltx2e}
\usepackage[cmex10]{amsmath}
\usepackage{array}
\usepackage{wasysym}
\usepackage{dsfont}
\usepackage[latin1]{inputenc}                       
\usepackage[english]{babel}                         
\usepackage[T1]{fontenc}
\usepackage{mathtools}
\usepackage{amssymb} 
\usepackage{enumerate}
\usepackage{bbm}
\usepackage{epsfig,syntonly}
\usepackage{verbatim,times}
\usepackage{epstopdf}
\usepackage{graphicx}
\usepackage{latexsym,fancyhdr,bm}
\usepackage[tight,footnotesize]{subfigure}
\usepackage{url}
\usepackage{bbm}
\usepackage{dsfont}
\usepackage{amsthm}

\newtheoremstyle{custom}
{} 
{} 
{} 
{} 
{\bfseries} 
{:} 
{.25em} 
{} 
\theoremstyle{custom}

\newtheorem{theorem}{Theorem}
\newtheorem{lemma}[theorem]{Lemma}

\newtheorem{remark}[theorem]{Remark}

\newtheorem*{theorem*}{Theorem}
\newtheorem*{lemma*}{Lemma}
\newtheorem*{proposition*}{Proposition}
\newtheorem*{definition*}{Definition}
\newtheorem*{example*}{Example}
\newtheorem*{remark*}{Remark}
\newtheorem*{corollary*}{Corollary}

\makeatletter
\let\l@ENGLISH\l@english
\makeatother

\title{Comparing the bit-MAP and block-MAP decoding thresholds of Reed-Muller codes on BMS channels}
\author{\IEEEauthorblockN{Shrinivas Kudekar\IEEEauthorrefmark{1}, Santhosh Kumar\IEEEauthorrefmark{2}, Marco Mondelli\IEEEauthorrefmark{3}, Henry D. Pfister\IEEEauthorrefmark{4}, R\"{u}diger Urbanke\IEEEauthorrefmark{3}, }
\IEEEauthorblockA{\IEEEauthorrefmark{1}Qualcomm Research, New Jersey, USA\\
Email: \texttt{skudekar@qti.qualcomm.com}}
\IEEEauthorblockA{\IEEEauthorrefmark{2}Department of Electrical and Computer Engineering, Texas A\&M University, College Station\\
Email: \texttt{santhosh.kumar@tamu.edu}}
\IEEEauthorblockA{\IEEEauthorrefmark{3}School of Computer and Communication Sciences, EPFL, Switzerland\\
Emails: \texttt{\{marco.mondelli, ruediger.urbanke\}@epfl.ch}}
\IEEEauthorblockA{\IEEEauthorrefmark{4}Department of Electrical and Computer Engineering, Duke University\\
Email: \texttt{henry.pfister@duke.edu}}
}

\begin{document}

\maketitle
\begin{abstract}
\noindent The question whether RM codes are capacity-achieving is a long-standing open problem in coding theory that was recently answered in the affirmative for transmission over erasure channels \cite{RMpaper-ITTran,RMpaper-STOC}. Remarkably, the proof does not rely on specific properties of RM codes, apart from their symmetry. Indeed, the main technical result consists in showing that any sequence of linear codes, with doubly-transitive permutation groups, achieves capacity on the memoryless erasure channel under bit-MAP decoding. Thus, a natural question is what happens under block-MAP decoding. In \cite{RMpaper-ITTran,RMpaper-STOC}, by exploiting further symmetries of the code, the bit-MAP threshold was shown to be sharp enough so that the block erasure probability also converges to 0. However, this technique relies heavily on the fact that the transmission is over an erasure channel.

We present an alternative approach to strengthen results regarding the bit-MAP threshold to block-MAP thresholds. This approach is based on a careful analysis of the weight distribution of RM codes. In particular, the flavor of the main result is the following: assume that the bit-MAP error probability decays as $N^{-\delta}$, for some $\delta>0$. Then, the block-MAP error probability also converges to 0. This technique applies to transmission over any binary memoryless symmetric channel. Thus, it can be thought of as a first step in extending the proof that RM codes are capacity-achieving to the general case.
\end{abstract}

\begin{IEEEkeywords}
RM codes, weight distribution, bit-MAP threshold, block-MAP threshold.
\end{IEEEkeywords}

\section{Introduction} \label{sec:intro}

Reed-Muller (RM) codes are among the oldest known codes: they were introduced by Muller in~\cite{Muller-ire54} and, soon after, Reed proposed a majority logic decoder in~\cite{Reed-ire54}. A binary Reed-Muller code RM$(n, v)$, parameterized by non-negative integers $n$ and $v$, is a linear code of length $N=2^n$ and rate $R=\frac1N \sum_{i=0}^v \binom{n}{i}$. It is well known that the minimum distance of this code is $2^{n-v}$~\cite{Macwilliams-1977}.

The idea that RM codes might achieve capacity appears to be rather old and has been discussed by a variety of authors (for a detailed list of references, see Section I-B of \cite{RMpaper-ITTran}). In particular, it was observed numerically that the block error probability under MAP decoding for transmission over the binary erasure channel (BEC) of short RM codes is very close to that of random codes~\cite{Carlet-isit05,Didier-it06,Mondelli-com14}. Beyond erasure channels, it was conjectured in~\cite{Costello-proc07} that the sequence of rate-$1/2$ self-dual RM codes achieves capacity on the binary-input AWGN channel. For rates approaching either $0$ or $1$ with sufficient speed, it was shown that RM codes can correct almost all erasure patterns up to the capacity limit\footnote{Some effort is required to define capacity for rates approaching $0$ or $1$. See \cite[Definition 16]{Abbe-it15} for details.} \cite{Abbe-it15,Abbe-stoc15}. For rates approaching $0$ fast enough, it was also proved that RM codes can correct random error patterns up to the capacity limit~\cite{Abbe-it15,Abbe-stoc15}. 

The conjecture that RM codes achieve capacity on erasure channels under MAP decoding, for any rate $R\in (0, 1)$, has been recently solved in \cite{RMpaper-STOC,RMpaper-ITTran}. Remarkably, the proof does not use any of the specific properties of RM codes, apart from their symmetry. Indeed, the main technical result consists in showing that any sequence of linear codes with a doubly transitive permutation group achieves capacity on a memoryless erasure channel under bit-MAP decoding. More specifically, the proof analyzes the extrinsic information transfer functions of the codes using the area theorem~\cite{Ashikhmin-it04,RU-2008} and the theory of symmetric monotone boolean functions~\cite{Friedgut-procams96}. Combining these two ingredients with the code symmetry, it is possible to show that a bit-MAP decoding threshold exists and that it coincides with channel capacity. Since RM codes are doubly transitive, they achieve capacity under bit-MAP decoding. In order to extend the result from the bit-MAP error probability to the block-MAP error probability, two strategies are considered in \cite{RMpaper-STOC,RMpaper-ITTran}.

Firstly, an upper bound on the bit-MAP error probability is combined with a lower bound on the minimum distance of the code. This strategy is successful for BCH codes: they are doubly transitive and, therefore, the bit erasure probability converges to $0$; they have a minimum distance which scales as $N/\log(N)$ and, therefore, the block erasure probability converges to $0$. However, the minimum distance of RM codes scales only as $\sqrt{N}$ and this does not suffice to prove directly the desired result.

Secondly, further symmetries of RM codes\footnote{The exact condition required is that the permutation group of the code contains a transitive subgroup isomorphic to ${\rm GL}(n,\mathbb{F}_2)$, the general linear group of degree $n$ over the Galois field $\mathbb{F}_2$.} are exploited in the framework of \cite{Bourgain-gafa97}. In this way, one can show that the bit-MAP threshold is sharp enough so that the block erasure probability converges to $0$. Note that this approach relies heavily on the fact that the transmission is over an erasure channel. 

In this paper, we present a new strategy to compare the bit-MAP and block-MAP decoding thresholds of RM codes. This strategy is based on the careful analysis of the weight distribution of the codes and, as such, it applies to the transmission over any binary memoryless symmetric (BMS) channel. In particular, the flavor of the main result is the following: assume that the bit-MAP error probability decays as $N^{-\delta}$, for some $\delta>0$; then, the block-MAP error probability goes to $0$.

The rest of the paper is organized as follows. Section \ref{sec:ubweight} states the main result and outlines its proof. Section \ref{sec:proofs} contains the proof of the intermediate lemmas and discusses some extensions. Section \ref{sec:concl} concludes the paper with some final remarks.

\section{Main Result} \label{sec:ubweight}

First, let us introduce some definitions. For $k\in \mathbb N$, denote $[k] \triangleq\{1,\cdots,k\}$.
It is well known that the codewords of RM$(n, v)$ are given by the evaluations of the polynomials in $n$ variables of degree at most $v$ over $\mathbb F_2$. With an abuse of notation, we can think of RM$(n, v)$ as the collection of such polynomials $f : {\mathbb F}_2^n \to {\mathbb F}_2$. The normalized weight of a function $f : {\mathbb F}_2^n \to {\mathbb F}_2$ is the normalized number of 1's in it, i.e.,
\begin{equation*}
{\rm wt}(f) = \frac{1}{2^n}|\{x\in {\mathbb F}_2^n : f(x)=1\}|.
\end{equation*}
The cumulative weight distribution of RM$(n, v)$ at a normalized weight $\alpha\in [0, 1]$ is denoted by $W_{n, v}(\alpha)$ and is defined as the number of codewords whose normalized weight is at most $\alpha$, i.e.,
\begin{equation}\label{eq:cwd}
W_{n, v}(\alpha) = |\{f\in \text{RM}(n, v) : {\rm wt}(f)\le \alpha\}|.
\end{equation}

Let us now define exactly the difference between bit-MAP and block-MAP decoding:
\begin{itemize}

	\item The \emph{bit-MAP decoder} outputs the most likely bit value for each position and its bit error probability is denoted by $P_{\rm b}$. 
	
	\item The \emph{block-MAP decoder} outputs the most likely codeword and its block error probability is denoted by $P_{\rm B}$. 

\end{itemize}

The main result of the paper is stated below. 

\begin{theorem}[From bit-MAP to block-MAP] \label{th:block}
Consider a sequence of RM$(n,v_n)$ codes with increasing $n$ and rate $R_n$ converging to $R \in (0,1)$. Assume that each code is transmitted over a BMS channel with Bhattacharyya parameter $z\in (0,1)$ and that the bit error probability of the bit-MAP decoder $P_{\rm b}$ is $O(2^{-n\delta})$, for some $\delta >0$. Then, the block error probability of the block-MAP decoder $P_{\rm B}$ tends to 0, as $n$ goes large. 
\end{theorem}

\begin{remark}[Case of the BEC]
Consider the special case of transmission over the BEC. Then, by applying Theorem \ref{th:block}, one obtains that RM codes achieve capacity under block-MAP decoding directly from the result on bit-MAP decoding (see eq. (10) after Theorem 20 in \cite{RMpaper-ITTran}), without resorting to the framework of \cite{Bourgain-gafa97}. 
\end{remark}

In order to prove the theorem above, it is useful to introduce a randomized version of the bit-MAP and block-MAP decoders:
\begin{itemize}
	\item The \emph{randomized bit-MAP decoder} outputs each bit value according to its posterior probability and its bit error probability is denoted by $P_{\rm b, r}$. 
	\item The \emph{randomized block-MAP decoder} outputs each codeword according to its posterior probability and its block error probability is denoted by $P_{\rm B, r}$. 
\end{itemize}

The error probabilities of the MAP decoders are related to the error probabilities of their randomized counterparts by the following lemma, which is proved in Section \ref{sec:proofs}. 

\begin{lemma}[MAP vs. randomized MAP]\label{lm:randomized}
Consider transmission of a code $\mathcal{C}$ over a BMS channel and let $P_{\rm b}$, $P_{\rm B}$, $P_{\rm b, r}$ and $P_{\rm B, r}$ be the error probabilities of the bit-MAP, block-MAP, randomized bit-MAP and randomized block-MAP decoders. Then, the following inequalities hold: 
\begin{equation}\label{eq:rand1}
P_{\rm b} \le P_{\rm b, r} \le 2\cdot P_{\rm b},
\end{equation}
\begin{equation}\label{eq:rand2}
P_{\rm B} \le P_{\rm B, r} \le 2\cdot P_{\rm B}.
\end{equation}
\end{lemma}

A crucial point in the proof of our main result is that for any RM code of sufficiently large block length and any $\beta>0$, the codewords at distance at most $2^{n(1-\beta)}$ from the transmitted codeword have a negligible effect on the block error probability under randomized block-MAP decoding. 

\begin{lemma}[Small distances do not count]\label{lemma:codeword}
Consider a sequence of RM$(n,v_n)$ codes with increasing $n$ and rate $R_n$ converging to $R \in (0,1)$. Assume that each code is transmitted over a BMS channel with Bhattacharyya parameter $z\in (0,1)$. Fix any $\beta > 0$. Then, the probability that the randomized block-MAP decoder outputs an incorrect codeword at Hamming distance at most $2^{n(1-\beta)}$ from the transmitted codeword tends to $0$ as $n$ tends to infinity. 
\end{lemma}

The proof of Lemma \ref{lemma:codeword} can be found in Section \ref{sec:proofs} and it relies on the following upper bound on the weight distribution of RM codes.

\begin{lemma}[Upper bound on weight distribution]\label{lemma:ubw}
Consider the code RM$(n, v)$. Pick an integer $\ell \in [v-1]$ and $\varepsilon\in (0, 1/2]$. Set $\alpha = 2^{-\ell}(1-\varepsilon)$. Then,
\begin{equation*}
W_{n, v}(\alpha) \le (1/\varepsilon)^{ c(v+2)^2(n\ell+ \sum_{i=0}^{v-\ell} \binom{n-\ell}{i})},
\end{equation*} 
for some universal constant $c\in \mathbb R$ which does not depend on any other parameter. 
\end{lemma}

The study of the weight distribution of RM codes is a classical problem in coding theory~\cite{Sloane-it70,Kasami-it70,Kasami-ic76}, which culminated in obtaining asymptotically tight bounds for fixed order $v$ and asymptotic $n$~\cite{Kaufman-it12}. These bounds were further refined in \cite{Abbe-it15,Abbe-stoc15}. In particular, Lemma \ref{lemma:ubw} is an improvement of the result stated in~\cite[Theorem 3.1]{Kaufman-it12} and it is proved in Section \ref{sec:proofs}. Finally, we can proceed with the proof of Theorem \ref{th:block}.

\begin{proof}[Proof of Theorem \ref{th:block}]
Let $P_{\rm B, r}^{\rm l}$ be the probability that the randomized block-MAP decoder outputs an incorrect codeword whose Hamming distance from the transmitted codeword is at most $2^{n(1-\delta/2)}$. Similarly, let $P_{\rm B, r}^{\rm h}$ be the probability that the randomized block-MAP decoder outputs an incorrect codeword whose Hamming distance from the transmitted codeword is at least $2^{n(1-\delta/2)}$. Then,
\begin{equation*}
P_{\rm B} \le P_{\rm B, r} = P_{\rm B, r}^{\rm l} + P_{\rm B, r}^{\rm h}, 
\end{equation*}
where the inequality comes from \eqref{eq:rand2}. By Lemma \ref{lemma:codeword}, we have that $P_{\rm B, r}^{\rm l}$ tends to $0$ as $n$ goes large. Hence, in order to prove the claim, it suffices to show that $P_{\rm B, r}^{\rm h}$ tends to $0$ as $n$ goes large.

To do so, we first upper bound $P_{\rm B, r}^{\rm h}$ as a function of $P_{\rm b, r}$, by adapting the proof of (13.51) in \cite[pp. 225-226]{MacKay-2003}. Let $x=(x_1, \cdots, x_N)$ denote a codeword and $y$ the channel output. By definition, the randomized bit-MAP decoder outputs the value in position $i$ according to the distribution $p(x_i\mid y)$. However, we can draw a sample from $p(x_i\mid y)$ also by first sampling from the joint distribution $p(x\mid y)$ and then discarding all positions except position $i$. Now, let $A$ be the event in which the Hamming distance between $x$ and the transmitted codeword is at least $2^{n(1-\delta/2)}$. Thus, 
\begin{equation}\label{eq:boundpbr}
\begin{split}
P_{\rm b, r} &= {\mathbb P}( A )\cdot  {\mathbb P}(\mbox{bit error}\mid A)+{\mathbb P}( A^{\rm c} )\cdot  {\mathbb P}(\mbox{bit error}\mid A^{\rm c})\\
&\ge  {\mathbb P}( A )\cdot  {\mathbb P}(\mbox{bit error}\mid A) \ge P_{\rm B, r}^{\rm h}\cdot 2^{-n\delta/2},
\end{split}
\end{equation}
where $A^{\rm c}$ denotes the complement of $A$. To prove the last inequality, note that ${\mathbb P}( A ) = P_{\rm B, r}^{\rm h}$ and that, since $x$ has Hamming distance at least $2^{n(1-\delta/2)}$ from the transmitted codeword, at least a fraction $2^{-n\delta/2}$ of the bits in $x$ is decoded incorrectly by the randomized bit-MAP decoder. Finally,  
\begin{equation*}
P_{\rm B, r}^{\rm h} \stackrel{\mathclap{\mbox{\footnotesize(a)}}}{\le} P_{\rm b, r}\cdot 2^{n\delta/2}\stackrel{\mathclap{\mbox{\footnotesize(b)}}}{\le} 2\cdot P_{\rm b}\cdot 2^{n\delta/2},
\end{equation*}
where (a) is obtained from \eqref{eq:boundpbr} and (b) from \eqref{eq:rand1}. Since, by hypothesis, $P_{\rm b}$ is $O(2^{-n\delta})$, the result is readily proved.
\end{proof}

\section{Proof of Lemmas and Extensions} \label{sec:proofs}

We start by proving Lemma \ref{lm:randomized}.

\begin{proof}[Proof of Lemma \ref{lm:randomized}]
The inequalities $P_{\rm b} \le P_{\rm b, r}$ and $P_{\rm B} \le P_{\rm B, r}$ follow from the fact that the MAP decoder is, by definition, an optimal decoder in the sense that it minimizes the error probability.

In order to prove the other inequality in \eqref{eq:rand2}, let $x\in \mathcal{C}$ denote a codeword, $y\in \mathcal Y$ the channel output, and $\hat{x}_{\rm B}(y)$ the estimate provided by the block-MAP decoder given the channel output $y$. Denote by $\mathds{1}(\cdot)$ the indicator function of an event. Then, we can rewrite $P_{\rm B}$ as
\begin{equation}\label{eq:compPB}
\begin{split}
P_{\rm B} &= \sum_{x\in \mathcal{C}} p(x) \sum_{y\in \mathcal{Y}} p(y\mid x)\mathds{1}(\hat{x}_{\rm B}(y)\neq x)\\
&=\sum_{y\in \mathcal{Y}} p(y)\sum_{x\in \mathcal{C}} p(x\mid y)\mathds{1}(\hat{x}_{\rm B}(y)\neq x)\\
&\stackrel{\mathclap{\mbox{\footnotesize(a)}}}{=}\sum_{y\in \mathcal{Y}} p(y)\sum_{x\in \mathcal{C}\setminus \hat{x}_{\rm B}(y)} p(x\mid y)\\
&=\sum_{y\in \mathcal{Y}} p(y)\left(1-p(\hat{x}_{\rm B}(y) \mid y)\right)\\
&\stackrel{\mathclap{\mbox{\footnotesize(b)}}}{=}1-\sum_{y\in \mathcal{Y}} p(y)\cdot p(\hat{x}_{\rm B}(y) \mid y),
\end{split}
\end{equation}
where in (a) we use that the estimate $\hat{x}_{\rm B}(y)$ provided by the block-MAP decoder is equal to a fixed codeword (more specifically, to the most likely one) with probability $1$, and in (b) we use that $\sum_{y\in \mathcal{Y}} p(y)=1$.  

Similarly, let $\hat{x}_{\rm B, r}(y)$ be the estimate provided by the randomized block-MAP decoder given the channel output $y$. Then, the following chain of inequalities holds
\begin{equation}\label{eq:compPB2}
\begin{split}
P_{\rm B, r} &= \sum_{x\in \mathcal{C}} p(x) \sum_{y\in \mathcal{Y}} p(y\mid x)\mathds{1}(\hat{x}_{\rm B, r}(y)\neq x)\\
&=\sum_{y\in \mathcal{Y}} p(y)\sum_{x\in \mathcal{C}} p(x\mid y)\mathds{1}(\hat{x}_{\rm B, r}(y)\neq x)\\
&\stackrel{\mathclap{\mbox{\footnotesize(a)}}}{=}\sum_{y\in \mathcal{Y}} p(y)\sum_{x\in \mathcal{C}} p(x\mid y)(1-p(x\mid y))\\
&\stackrel{\mathclap{\mbox{\footnotesize(b)}}}{\le}\sum_{y\in \mathcal{Y}} p(y)\left(1-\left(\max_{x\in \mathcal{C}}p(x \mid y)\right)^2\right)\\
&=\sum_{y\in \mathcal{Y}} p(y)\left(1-p(\hat{x}_{\rm B}(y) \mid y)^2\right)\\
& \stackrel{\mathclap{\mbox{\footnotesize(c)}}}{=} 1-\sum_{y\in \mathcal{Y}} p(y)\cdot p(\hat{x}_{\rm B}(y) \mid y)^2\\
&\stackrel{\mathclap{\mbox{\footnotesize(d)}}}{\le} 1-\left(\sum_{y\in \mathcal{Y}} p(y)\cdot p(\hat{x}_{\rm B}(y) \mid y)\right)^2\\
& \stackrel{\mathclap{\mbox{\footnotesize(e)}}}{=} 1- (1- P_{\rm B})^2 \le 2\cdot P_{\rm B}.
\end{split}
\end{equation}
To prove the equality (a), we use that the estimate $\hat{x}_{\rm B, r}(y)$ provided by the randomized block-MAP decoder is equal to $x$ with probability $p(x\mid y)$. To prove inequality (b), we use that, given $m$ real numbers $1\ge p_1\ge \cdots \ge p_m\ge 0$ with $\sum_{j=1}^m p_j=1$, then 
\begin{equation*}
\begin{split}
\sum_{j=1}^m p_j(1-p_j) &= (1-p_1)\sum_{j=1}^m p_j\frac{1-p_j}{1-p_1}\\
&= (1-p_1)\left(p_1 + \sum_{j=2}^m p_j\frac{1-p_j}{1-p_1}\right)\\
&\le (1-p_1)\left(p_1 + \sum_{j=2}^m p_j\frac{1}{1-p_1}\right)\\
&= (1-p_1)\left(p_1 + 1\right) = 1-p_1^2.\\
\end{split}
\end{equation*}
To prove the equality (c), we use that $\sum_{y\in \mathcal{Y}} p(y)=1$. Inequality (d) follows from Jensen's inequality and equality (e) uses \eqref{eq:compPB}.

In order to prove the analogous inequality $P_{\rm b, r}\le 2\cdot P_{\rm b}$, one possibility is to write expressions similar to \eqref{eq:compPB} and \eqref{eq:compPB2} for the bit error probability of position $i$ under bit-MAP decoding and under randomized bit-MAP decoding, respectively. Otherwise, one can follow the simpler argument of \cite[p. 225]{MacKay-2003}, which we reproduce here for the sake of completeness. 

Consider first the case of a single bit with posterior probability $\{p_0, p_1\}$. Then,
$P_{\rm b} = \min(p_0, p_1)$. Furthermore, the probability that the randomized bit-MAP decoder makes a correct decision is $p_0^2 + p_1^2$, since its output and the ground truth follow the same distribution. Thus,
\begin{equation}\label{eq:ineqbit}
P_{\rm b, r} = 2p_0p_1 \le 2\min(p_0, p_1) = 2\cdot P_{\rm b}.
\end{equation}
In general, $P_{\rm b}$ and $P_{\rm b, r}$ are just the averages of many such error probabilities. Therefore, \eqref{eq:ineqbit} holds for the transmission of any number of bits.
\end{proof}

Now, let us state some more definitions and intermediate results which will be useful in the following.

Let $f : {\mathbb F}_2^n \to {\mathbb F}_2$ be a function. The derivative of $f$ in direction $y\in {\mathbb F}_2^n$ is denoted by $\Delta_y f: {\mathbb F}_2^n \to {\mathbb F}_2$ and it is defined as
\begin{equation*}
\Delta_y f(x) = f(x+y)+f(x).
\end{equation*}
Similarly, the $k$-iterated derivative of $f$ in directions $Y=(y_1, \cdots, y_k) \in ({\mathbb F}_2^n)^k$ is denoted by $\Delta_Y f: {\mathbb F}_2^n \to {\mathbb F}_2$ and it is defined as 
\begin{equation*}
\Delta_Y f(x) = \Delta_{y_1}\Delta_{y_2}\cdots\Delta_{y_k} f(x).
\end{equation*}
A simple manipulation yields 
\begin{equation*}
\Delta_Y f(x) = \sum_{I\subseteq [k]} f\left(x+\sum_{i\in I}y_i \right), 
\end{equation*}
from which it is clear that the order of $y_1, \cdots, y_k$ is irrelevant in the computation of $\Delta_Y f(x)$ and, therefore, we can think of $Y$ as a multi-set of size $k$.

Note that if $f$ is a degree $v$ polynomial, then its derivatives have degree at most $v-1$ and, consequently, its $k$-iterated derivatives have degree at most $d-k$. Furthermore, as pointed out in \cite[Section III]{Abbe-it15}, we have that $\Delta_y f(x) = \Delta_y f(x+y)$. Thus, in general, $\Delta_Y f(x)$ is determined by its values on the quotient space ${\mathbb F}_2^n \setminus \langle Y \rangle$, where $\langle Y \rangle$ denotes the space spanned by the vectors in $Y$.

The following lemma plays a central role in the proof of the upper bound on the weight distribution.

\begin{lemma}[Lemma 2.1 in \cite{Kaufman-it12}]\label{lemma:algo}
Pick an integer $\ell \ge 1$ and $\varepsilon\in (0, 1)$. Consider a function $f : {\mathbb F}_2^n \to {\mathbb F}_2$ s.t. ${\rm wt}(f)\le 2^{-\ell}(1-\varepsilon)$. Pick any $\delta > 0$. There exists a universal algorithm $\mathcal A$ (which does not depend on $f$) with the following properties:
\begin{enumerate}
\item $\mathcal A$ has two inputs: $x\in {\mathbb F}_2^n$ and $Y_1, \cdots, Y_t \in ({\mathbb F}_2^n)^{\ell}$.
\item $\mathcal A$ has oracle access to the $\ell$-iterated derivatives $\Delta_{Y_1}f, \cdots, \Delta_{Y_t}f$. 
\end{enumerate}
Then, for $t\le c (\log_2(1/\delta)\log_2(1/\varepsilon)+\log_2(1/\delta)^2)$,
where $c$ is a universal constant, there exists a choice of $Y_1, \cdots, Y_t$ s.t. 
\begin{equation}
{\mathbb P}({\mathcal A}(x; Y_1, \cdots, Y_t, \Delta_{Y_1}f, \cdots, \Delta_{Y_t}f) = f(x))\ge 1-\delta,
\end{equation}
where the probability distribution is over $x\in \mathbb F_2^n$ chosen uniformly at random.
\end{lemma} 

In words, Lemma \ref{lemma:algo} says that any function of small normalized weight can be approximated arbitrarily well, given a sufficient amount of its derivatives. For a proof of this result, we refer the interested reader to \cite[Section II]{Kaufman-it12}. 

At this point we are ready to prove Lemma \ref{lemma:ubw}.

\begin{proof}[Proof of Lemma \ref{lemma:ubw}]
Pick $\delta = 2^{-v-1}$. Apply the universal algorithm $\mathcal A$ to all the codewords $f\in \text{RM}(n, v)$. Denote by $\mathcal H$ the family of functions obtained by doing so. In other words, $\mathcal H$ is the set of outputs of $\mathcal A$ when the input is a degree $v$ polynomial in $n$ variables. 

By Lemma \ref{lemma:algo}, for any $f\in \text{RM}(n, v)$ s.t. ${\rm wt}(f)\le \alpha$, there exists $h\in \mathcal H$ which differs from $f$ in a fraction $<\delta$ of points of ${\mathbb F}_2^n$.

Suppose now that there exists $h\in \mathcal H$ which is obtained by applying the algorithm $\mathcal A$ to two distinct codewords $f_1, f_2 \in$ RM$(n,v)$ s.t. ${\rm wt}(f_1)\le \alpha$ and ${\rm wt}(f_2)\le \alpha$. Then, $h$ differs from $f_1$ in a fraction $<\delta$ of points and $h$ differs from $f_2$ in a fraction $<\delta$ of points. Therefore, $f_1$ and $f_2$ can differ in a fraction $<2\delta=2^{-v}$ of points. As the minimum distance of the code is $2^{n-v}$, we conclude that $f_1=f_2$, and, consequently, we can associate to each $f\in \text{RM}(n, v)$ s.t. ${\rm wt}(f)\le \alpha$ a unique $h\in \mathcal H$. This implies that
\begin{equation}\label{eq:boundcardH}
W_{n, v}(\alpha) \le |\mathcal H|.
\end{equation}
The remainder of the proof consists in upper bounding the cardinality of $\mathcal H$.

Recall that the algorithm $\mathcal A$ takes as input:
\begin{enumerate}
\item the $t$ directions $Y_1, \cdots, Y_t \in ({\mathbb F}_2^{n})^{\ell}$ with $t\le c(\log_2(1/\delta)\log_2(1/\varepsilon)+\log_2(1/\delta)^2)$,
\item the $t$ $\ell$-iterated derivatives of the input.
\end{enumerate}
The number of different possibilities for each $Y_i$ (with $i\in [t]$) is $2^{n\ell}$. Given $Y_i$, the number of possible functions $\Delta_{Y_i} f$ is upper bounded by the number of polynomials of degree at most $v-\ell$ defined in the space ${\mathbb F}_2^n \setminus \langle Y_i \rangle$. As this space has dimension $n-\ell$, the number of possible functions $\Delta_{Y_i} f$ is $2^{\sum_{j=0}^{v-\ell} \binom{n-\ell}{j}}$. 

By putting everything together, we conclude that
\begin{equation*}
\begin{split}
|\mathcal H| &\le 2^{t(n\ell+\sum_{j=0}^{v-\ell} \binom{n-\ell}{j})} \\
&\stackrel{\mathclap{\mbox{\footnotesize(a)}}}{\le} 2^{c((v+1)\cdot\log_2(1/\varepsilon)+(v+1)^2)(n\ell+\sum_{j=0}^{v-\ell} \binom{n-\ell}{j})}\\
&\stackrel{\mathclap{\mbox{\footnotesize(b)}}}{\le} 2^{c((v+1)\cdot\log_2(1/\varepsilon)+(v+1)^2\cdot\log_2(1/\varepsilon))(n\ell+\sum_{j=0}^{v-\ell} \binom{n-\ell}{j})}\\
&\le (1/\varepsilon)^{c(v+1)(v+2)(n\ell+\sum_{j=0}^{v-\ell} \binom{n-\ell}{j})}\\
&\le (1/\varepsilon)^{c(v+2)^2(n\ell+\sum_{j=0}^{v-\ell} \binom{n-\ell}{j})},
\end{split}
\end{equation*}
where (a) combines the upper bound on $t$ with the choice $\delta = 2^{-v-1}$, and (b) uses that $\log_2(1/\varepsilon)\ge 1$ for $\varepsilon \in (0, 1/2]$. 
\end{proof} 

As it was pointed out in Section \ref{sec:ubweight}, Lemma \ref{lemma:ubw} is a refinement of~\cite[Theorem 3.1]{Kaufman-it12}. More specifically, the upper bound \eqref{eq:boundcardH} comes from the proof of~\cite[Theorem 3.1]{Kaufman-it12} and our refinement consists in a more accurate upper bound on $|\mathcal{H}|$. Note that this improvement is necessary to obtain the desired result on the error probability of the randomized block-MAP decoder, as the upper bound on the weight distribution presented in~\cite[Theorem 3.1]{Kaufman-it12} is not tight enough for this purpose.

Let us proceed with the proof of Lemma \ref{lemma:codeword}.

\begin{proof}[Proof of Lemma \ref{lemma:codeword}]
Since the claim to be proved is stronger when $\beta$ is smaller, we can assume without loss of generality that $\beta \in (0, 1/2)$. Suppose now that, for $n$ large enough,
\begin{equation} \label{eq:condr}
v_n > \frac{n}{2}\left(1+\frac{\beta}{2}\right).
\end{equation}
Then,
\begin{equation*}
\begin{split}
R_n &= \frac{1}{2^n} \sum_{i=0}^{v_n} \binom{n}{i} = 1-\frac{1}{2^n} \sum_{i=0}^{n-v_n-1} \binom{n}{i} \\
&\ge 1-\frac{2^{nh_2(\frac{n-v_n-1}{n})}}{2^n},
\end{split}
\end{equation*}
with $h_2(x)=-x\log_2(x)-(1-x)\log_2(1-x)$ and where the last inequality is an application of \cite[Lemma 4.7.2]{Ash-1990} (or, equivalently, of \cite[Eqn.~(1.59)]{RU-2008}) as $n-v_n-1 \le n/2$. This means that, for any $\beta\in (0, 1/2)$, if $v_n$ satisfies \eqref{eq:condr}, then the rate $R_n$ tends to 1. Similarly, it is easy to see that if $v_n < \displaystyle\frac{n}{2}\left(1-\frac{\beta}{2}\right)$, then the rate $R_n$ tends to 0. Since $R_n$ converges to $R\in (0, 1)$, we have that, for $n$ large enough,
\begin{equation} \label{eq:condrnew}
v_n \in \left( \frac{n}{2}\left(1-\frac{\beta}{2}\right), \frac{n}{2}\left(1+\frac{\beta}{2}\right)\right).
\end{equation}

Let $x$ denote a codeword and $y$ the channel output. Then, the posterior probability $p(x\mid y)$ can be written as
\begin{equation}\label{eq:posterior}
p(x\mid y) = \frac{p(y\mid x)p(x)}{\sum_{\tilde{x}}p(y\mid \tilde{x})p(\tilde{x})} = \frac{p(y\mid x)}{\sum_{\tilde{x}}p(y\mid \tilde{x})},
\end{equation}
where the last equality comes from the fact that the codeword is chosen uniformly from the codebook. From \eqref{eq:posterior} we deduce that, by adding codewords, the posterior probability $p(x\mid y)$ decreases. Then, the probability that the randomized block-MAP decoder outputs a specific codeword $x$ increases if we remove all codewords except $x$ and the codeword that was actually transmitted. By using \eqref{eq:rand2}, we can upper bound such a probability by 2 times the block error probability of the non-randomized block-MAP decoder. Eventually, this last probability is upper bounded by $\frac{1}{2}z^{w}$, where $w$ is the Hamming weight of $x$~\cite[Lemma 4.67]{RU-2008}.

The argument above proves that the probability that the randomized block-MAP decoder outputs a codeword of weight $w \in [2^n]$ is upper bounded by $z^{w}$. Hence, by applying the union bound, the probability that the randomized block-MAP decoder outputs a codeword of normalized weight at most $2^{-n\beta}$ is upper bounded by
\begin{equation*}
\sum_{w = 1}^{\lceil 2^{n(1-\beta)}\rceil} z^{w} c_w,
\end{equation*}
where $c_w$ denotes the number of codewords of weight $w$. As the minimum distance of the code RM$(n, v_n)$ is $2^{n-v_n}$, we deduce that $c_w =0$ for $w\in \{1, \cdots, 2^{n-v_n}-1\}$. For $w \in \{2^{n-v_n}, \cdots, \lceil 2^{n(1-\beta)}\rceil\}$, we have that
\begin{equation}\label{eq:boundcw}
\begin{split}
\log_2 & (c_w)  \stackrel{\mathclap{\mbox{\footnotesize(a)}}}{\le} \log_2 \left(W_{n, v_n}(w2^{-n})\right)\\ &\stackrel{\mathclap{\mbox{\footnotesize(b)}}}{\le} \log_2\left(W_{n, v_n}(2^{\lceil\log_2(w)\rceil-n})\right)\\
&\stackrel{\mathclap{\mbox{\footnotesize(c)}}}{\le} c(v_n+2)^2\Biggl(n(n-\lceil \log_2(w) \rceil -1)\\
&+ \sum_{i=0}^{v_n -n+\lceil \log_2(w) \rceil +1} \binom{\lceil \log_2(w) \rceil +1}{i}\Biggr)\\
&\stackrel{\mathclap{\mbox{\footnotesize(d)}}}{\le} cn^2\left(n^2+ 2^{(\lceil \log_2(w) \rceil +1)\cdot h_2\left(\frac{v_n -n+\lceil \log_2(w) \rceil +1}{\lceil \log_2(w) \rceil +1}\right)}\right)\\
&\stackrel{\mathclap{\mbox{\footnotesize(e)}}}{\le} cn^2\left(n^2+ 2^{( \log_2(w) +2)\cdot h_2\left(\frac{n\beta/4 -n/2+ \log_2(w) +2}{ \log_2(w)}\right)}\right),
\end{split}
\end{equation}
where (a) comes from the definition \eqref{eq:cwd} of cumulative weight distribution, (b) comes from the fact that $W_{n, v_n}(\alpha)$ is increasing in $\alpha$, (c) comes from the application of Lemma \ref{lemma:ubw} with $\ell=n-\lceil \log_2(w) \rceil -1$ and $\varepsilon=1/2$, (d) comes from the application of \cite[Lemma 4.7.2]{Ash-1990} (or, equivalently, of \cite[Eqn.~(1.59)]{RU-2008}), and (e) comes from the fact that $h_2(x)$ is increasing for $x\in [0, 1/2]$ and $v_n$ is upper bounded by \eqref{eq:condrnew}. Note that we fulfill the hypotheses of Lemma \ref{lemma:ubw} since $w \ge 2^{n-v_n}$ implies that $\ell \le v_n-1$. In addition, we can apply \cite[Lemma 4.7.2]{Ash-1990} since \eqref{eq:condrnew} and $w\le \lceil 2^{n(1-\beta)}\rceil $ imply that $v_n -n+\lceil \log_2(w) \rceil +1 \le \frac{\lceil \log_2(w) \rceil +1}{2}$ for $n$ large enough.

Thus, the logarithm of the desired probability is upper bounded as follows:
\begin{equation}\label{eq:logbound}
\begin{split}
&\log_2  \Biggl(\sum_{w = 1}^{\lceil 2^{n(1-\beta)}\rceil} z^{w} c_w\Biggr) \stackrel{\mathclap{\mbox{\footnotesize(a)}}}{=} \log_2 \Biggl(\sum_{w = 2^{n-v_n}}^{\lceil 2^{n(1-\beta)}\rceil} z^{w}c_w \Biggr)\\
&\stackrel{\mathclap{\mbox{\footnotesize(b)}}}{\le} n +\max_{w \in \mathbb N \cap [2^{n-v_n}, \lceil 2^{n(1-\beta)}\rceil]}\log_2 (z^{w} c_w) \\
&\le n + \max_{w \in [2^{n-v_n}, \lceil 2^{n(1-\beta)}\rceil]}\log_2(z^{w} c_w )\\
&\stackrel{\mathclap{\mbox{\footnotesize(c)}}}{\le} n + cn^4 + \max_{\log_2(w) \in [n-v_n, n(1-\beta)+1]} \Biggl(-\log_2(1/z)\cdot 2^{\log_2(w)} \\
&+ cn^2 2^{( \log_2(w) +2)\cdot h_2\left(\frac{n\beta/4 -n/2+ \log_2(w) +2}{ \log_2(w)}\right)}\Biggr)\\
&\stackrel{\mathclap{\mbox{\footnotesize(d)}}}{\le} n + cn^4 + \max_{x \in [1/2-\beta/4, (1-\beta)+1/n]} \Biggl(-\log_2(1/z)\cdot 2^{nx} \\
&+ cn^2 2^{n( x +2/n)\cdot h_2\left(\frac{\beta/4 -1/2+ x +2/n}{ x}\right)}\Biggr)\\
&\stackrel{\mathclap{\mbox{\footnotesize(e)}}}{\le} n + cn^4 + \max_{x \in [1/2-\beta/4, 1-7\beta/8]} \Biggl(-\log_2(1/z)\cdot 2^{nx} \\
&+ 4cn^2 2^{n x \cdot h_2\left(\frac{\beta/3 -1/2+ x}{ x}\right)}\Biggr),
\end{split}
\end{equation}
where (a) uses that $c_w=0$ for $w\in \{1, \cdots, 2^{n-v_n}-1\}$, (b) uses that the number of terms in the sum is upper bounded by $2^n$, (c) uses \eqref{eq:boundcw}, in (d) we set $x = \log_2(w)/n$ and we use the upper bound \eqref{eq:condrnew} on $v_n$, and in (e) we use that $h_2(t)\le 1$ for any $t\in [0, 1]$ and that $1-7\beta/8 \ge 1-\beta+1/n$ and $h_2\left(\frac{\beta/4 -1/2+ x +2/n}{ x}\right) \le h_2\left(\frac{\beta/3 -1/2+ x}{ x}\right)$ for $n$ large enough. 

In order to conclude, it suffices to observe that, for any $\beta \in (0,1/2)$ and any $x\in [1/2-\beta/4, 1-7/8\beta]$, we have
\begin{equation*}
h_2\left(\frac{\beta/3 -1/2+ x}{ x}\right) < 1,
\end{equation*}
which implies that the upper bound in \eqref{eq:logbound} tends to $-\infty$ and, therefore, the desired probability goes to 0. 
\end{proof}

The following two remarks discuss how to tighten the main result by making the hypothesis on the decay rate of $P_{\rm b}$ less restrictive and by evaluating the decay rate of $P_{\rm B}$. 

\begin{remark}[Looser condition on $P_{\rm b}$]\label{rmk:loose}
In order to have that $P_{\rm B}\to 0$, Theorem \ref{th:block} requires that $P_{\rm b}$ is $O(2^{-n\delta})$ for some $\delta>0$. With some more work, one can conclude that $P_{\rm B}\to 0$ even under the less restrictive hypothesis that $P_{\rm b}$ is $O(2^{-n^{1/2+\delta'}})$ for some $\delta'>0$. The proof of this tighter result is based on a stronger version of Lemma \ref{lemma:codeword} which is outlined in the next paragraph. 

Consider the same transmission scenario of Lemma \ref{lemma:codeword} and fix any $\beta' >0$. Then, the probability that the randomized block-MAP decoder outputs an incorrect codeword at Hamming distance at most $2^{n-n^{1/2+\beta'}}$ tends to $0$ as $n$ tends to infinity. In other words, codewords with distances up to $2^{n-n^{1/2+\beta'}}$, for any $\beta' >0$, do not count, as opposed to distances up to $2^{n(1-\beta)}$, for any $\beta>0$, in the original statement. In order to prove this stronger claim, first one needs this tighter bound for the range of $v_n$ (compare to \eqref{eq:condrnew}):
\begin{equation*}
v_n \in \left( \frac{n}{2}-\frac{n^{1/2+\beta'}}{4}, \frac{n}{2}+\frac{n^{1/2+\beta'}}{4}\right). 
\end{equation*}
Indeed, it follows from simple manipulations that $v_n > n/2+n^{1/2+\beta'}/4$ yields rates $R_n\to 1$ and $v_n < n/2-n^{1/2+\beta'}/4$ yields rates $R_n\to 0$. Then, one obtains this bound on $\log_2 (c_w)$  (compare to the last inequality in \eqref{eq:boundcw}):
\begin{equation*}
\log_2 (c_w) \le cn^2\hspace*{-0.2em}\left( \hspace*{-0.2em}n^2 \hspace*{-0.2em}+ \hspace*{-0.2em} 2^{( \log_2(w) +2) h_2\left(\frac{\frac{n^{1/2+\beta'}}{4} -\frac{n}{2}+ \log_2(w) +2}{ \log_2(w)}\right)}\right),
\end{equation*}
which yields the following upper bound on the logarithm of the desired probability (compare to the last inequality in \eqref{eq:logbound}),
\begin{equation}\label{eq:logboundref}
\begin{split}
&n + cn^4 + \max_{x \in [1/2-n^{\beta'-1/2}/4, 1-7n^{\beta'-1/2}/8]} \Biggl(-\log_2(1/z)\cdot 2^{nx} \\
&+ 4cn^2 2^{n x \cdot h_2\left(\frac{n^{\beta'-1/2}/3 -1/2+ x}{ x}\right)}\Biggr).
\end{split}
\end{equation}
Eventually, when $n\to \infty$, one can show that the above quantity tends to $-\infty$, which suffices to prove the claim. Note that the result of Theorem \ref{th:block} cannot be further improved by using a better upper bound on the weight distribution. Indeed, a simple counting argument gives $W_{n, v}(2^{-\ell}) \ge 2^{ n\ell+ \binom{n-\ell}{v-\ell}}$~\cite[Section III]{Abbe-it15}. 
\end{remark}

\begin{remark}[Decay rate of $P_{\rm B}$]
The decay rate of $P_{\rm B}$ is given by the slowest between the decay rates of $P_{\rm B, r}^{\rm l}$ and $P_{\rm B, r}^{\rm h}$, defined in the proof of Theorem \ref{th:block}.

First, assume that $P_{\rm b}$ is $O(2^{-n\delta})$, for some $\delta>0$, as in the hypothesis of the theorem. Note that \eqref{eq:logbound} is minimized when $x=1/2-\beta/4$ and we can pick any $\beta\le \delta/2$, since $\beta$ is set to $\delta/2$ in the proof of Theorem \ref{th:block} and the claim of Lemma \ref{lemma:codeword} is stronger when $\beta$ is smaller. Therefore, one obtains that $P_{\rm B, r}^{\rm l}$ is $O(2^{-2^{n\gamma}})$, for any $\gamma\in (0, 1/2)$. This bound essentially comes from the fact that the minimum distance of RM codes scales as $\sqrt{N}$. From the argument in the last paragraph of the proof of Theorem \ref{th:block}, one obtains that $P_{\rm B, r}^{\rm h}$ is $O(2^{-n\rho})$, for any $\rho\in (0, \delta)$. Thus, when $P_{\rm b}$ is $O(2^{-n\delta})$, we conclude that $P_{\rm B, r}^{\rm h}$ is $O(2^{-n\rho})$ for any fixed $\rho < \delta$. 

Now, assume that $P_{\rm b}$ is $O(2^{-n^{1/2+\delta'}})$, for some $\delta'>0$, as in Remark \ref{rmk:loose}. From \eqref{eq:logboundref}, one obtains again that $P_{\rm B, r}^{\rm l}$ is $O(2^{-2^{n\gamma}})$, for any $\gamma \in(0, 1/2)$. From the argument in the proof of Theorem \ref{th:block}, one obtains that $P_{\rm B, r}^{\rm h}$ is $O(2^{- a \cdot n^{1/2+\delta'}})$, for any $a\in (0, 1)$, which also gives the overall decay rate of $P_{\rm B}$.
In conclusion, these arguments show that the decay rates of $P_{\rm b}$ and $P_{\rm B}$ are essentially the same. 
\end{remark}

\section{Conclusions}\label{sec:concl}

In this paper, we propose a comparison between the bit-MAP and block-MAP decoding thresholds of Reed-Muller codes via the careful analysis of their weight distribution. In particular, we show that, if the bit error probability under bit-MAP decoding tends to $0$ with sufficient speed, then the block error probability under block-MAP decoding also tends to $0$. Specializing this result to the case of the BEC, we obtain an alternative proof of the fact that RM codes achieve capacity under block-MAP decoding, which does not rely on the framework in~\cite{Bourgain-gafa97}. Furthermore, since the main result of the paper applies to the transmission over any BMS channel, it could be seen as a first step towards the generalization of the ideas in \cite{RMpaper-ITTran,RMpaper-STOC} beyond the erasure channel.

\section*{Acknowledgement}

The work of S.~Kumar and H.~D.~Pfister was supported in part by the National Science Foundation (NSF) under Grant No. 1218398. The work of M.~Mondelli and R.~Urbanke was supported by grant No.\ 200020\_146832/1 of the Swiss National Science Foundation. M.~Mondelli was also supported by the Dan David Foundation.

\vspace*{-0.25em}


\bibliographystyle{IEEEtran}
\bibliography{../../bibtex/WCLabrv,../../bibtex/WCLbib,../../bibtex/WCLnewbib}

\end{document}